\documentclass[review]{elsarticle}

\usepackage{lineno,hyperref}
\usepackage{amsmath}
\usepackage{amssymb}
\usepackage{amsthm}
\theoremstyle{plain}
\newtheorem{thm}{Theorem}[section]
\theoremstyle{definition}
\newtheorem{defn}{Definition}[section]
\newtheorem{prop}{Proposition}[section]
\newtheorem{cor}{Corollary}[section]
\newtheorem{example}{Example}[section]

\begin{document}

\begin{frontmatter}

\title{Generalized twisted centralizer codes}

\author{Joydeb Pal}
\ead{joydebpal77@gmail.com}
\address{Department of Mathematics\\ 
         National Institute of Technology Durgapur\\
         Burdwan, India. \\
          }

\author{Pramod Kumar Maurya}
\ead{pramod$\_$kumar22490@hotmail.com}
\address{Department of Mathematics\\ 
         National Institute of Technology Durgapur\\
         Burdwan, India.}

\author{Shyambhu Mukherjee}
\ead{pakummukherjee@gmail.com}
\address{SMU Department,
         Indian Statistical Institute\\ 
         Bangalore, Karnataka, India.}
          
\author{Satya Bagchi} 
\ead{satya.bagchi@maths.nitdgp.ac.in}
\address{Department of Mathematics\\ 
         National Institute of Technology Durgapur\\
         Burdwan, India.}

\begin{abstract}
Centralizer codes of length $n^2$ is obtained by taking centralizer of a square matrix over a finite field $ \mathbb{F}_q $. Twisted centralizer codes, twisted by an element $ a \in \mathbb{F}_q $, are also similar type of codes but different in nature. The main results of these codes were embedded on dimension and minimum distance. In this paper, we have defined a new family of twisted centralizer codes namely generalized twisted centralizer (GTC) codes by $\mathcal{C}(A,D):= \lbrace B \in \mathbb{F}_q^{n \times n}|AB=BAD \rbrace$ twisted by a matrix $D$ and investigated results on dimension and minimum distance. Parity-check matrix and syndromes are also investigated. Length of the centralizer codes is $ n^2 $ by construction but in this paper, we have constructed centralizer codes of length $ (n^2-i) $, where $ i $ is a positive integer. In twisted centralizer codes, minimum distance can be at most $ n $ when the field is binary whereas GTC codes can be constructed with minimum distance more than $ n $.
\end{abstract}
\end{frontmatter}

\section{Introduction}
Let $ \mathbb{F}_q $ be a finite field with $ q $ elements. The set of all square matrices of order $ n $ over $ \mathbb{F}_q $ is denoted by $ \mathbb{F}_q^{ n \times n } $.

Algebraic codes are important tools in data transmission. Ability of a good code is that, it detects or corrects more errors of an  encoded message when it is transmitted over a noisy channel. An error correction capability of a code totally depends on its construction. A fundamental problem in error correcting codes is to produce a code $ [n, k, d] $ with given $n$ and $k$, find maximum possible minimum distance $d$.

Centralizer codes are very special type linear codes of length $ n^2 $. The concept of the centralizer codes are beautifully constructed in \cite{Alahmadi201468}. For $A \in \mathbb{F}_q^{n \times n}$, the centralizer code is defined by $ \mathcal{C}(A):= \lbrace B \in \mathbb{F}_q^{n \times n}|AB=BA\rbrace $. The authors have computed bounds on dimension. They have given an efficient encoding and decoding procedure. It has shown that centralizer codes can locate a single error by looking at syndrome only. If $ A $ is a non cyclic matrix then centralizer code $ \mathcal{C}(A) $  has dimension greater than $n$. Non cyclic matrices are very rare according to \cite{JLMS} but the adjacency matrices of distance regular graphs of diameter less than $ n-1 $, are not cyclic. Thus authors relates automorphism groups of graphs with centralizer codes.

In 2017, this work is extended in \cite{Alahmadi2017235}, namely twisted centralizer codes, defined as $\mathcal{C}(A,a):= \lbrace B \in \mathbb{F}_q^{n \times n}|AB=aBA \rbrace$, where $a \in \mathbb{F}_q$. It is clear from the definition that centralizer codes are special kinds of twisted centralizer codes for $ a = 1 $. It has been shown that dimensions of centralizer codes and twisted centralizer codes are equal if there is an invertible matrix in the code. They have refined bounds of dimension and minimum distance in centralizer codes. These codes have less computational complexity to decode a received codeword through noisy channel. They have ability to correct single error only and also assert that if $ a \neq 0, 1 $ then the minimal distance can be greater than $ n $ whereas in centralizer codes (for $ a = 1 $) the minimal distance is at most $ n $.

In this paper, we define generalized twisted centralizer (GTC) codes, obtained from $A$ twisted by a matrix $D \in \mathbb{F}_q^{n \times n}$, defined as  $ \mathcal{C}(A,D):= \lbrace B \in \mathbb{F}_q^{n \times n}|AB=BAD \rbrace $. It is clear that $\mathcal{C}(A,D)$ is a $\mathbb{F}$-linear subspace of the vector space $\mathbb{F}_q^{n \times n}$. The centralizer codes defined in \cite{Alahmadi201468} are obtained from $\mathcal{C}(A,D)$ when $D = I_n$, identity matrix of order $n$ and twisted centralizer codes defined in \cite{Alahmadi2017235} are obtained from $\mathcal{C}(A,D)$ when $D = aI_n$, scalar matrix of order $n$. $\mathcal{C}(A,D)$ is considered to be a code by constructing codewords of length $n^2$ from matrices $B \in \mathcal{C}(A,D)$ by writing column-by-column. We execute some salient results of twisted centralizer codes. We give some idea on centralizer code of various length which is not of the form $n^2$ using the concept of puncture codes. Some examples are given which are the witness on existence of generalized twisted centralizer codes. We show that for a matrix $D \in \mathbb{F}_2^{ n \times n}$ minimum distance of GTC codes can be larger than $ n $.

The paper is organized as follows. In Section $2$, we give definition of GTC code and establish some basic results on parity check matrix and dimension. In Section $3$, we explain our main results. Complete encoding and decoding procedure is discussed in Section $4$. In Section $5$, we provide some examples on optimal GTC codes. We provide GTC codes of length less than $ n^2 $ in Section $6$. In Section $7$, we give conclusion with an open problem.

\section{Preliminaries}
\begin{defn}
For any square matrix $A \in \mathbb{F}_q^{n \times n}$ and any matrix $D \in \mathbb{F}_q^{n \times n}$, the subspace $\mathcal{C}(A,D):= \lbrace B \in \mathbb{F}_q^{n \times n}|AB=BAD \rbrace $ of $\mathbb{F}_q^{n \times n}$ is called generalized twisted centralizer code of $ A $ twisted by the matrix $D$.
\end{defn}

\begin{prop}
Parity-check matrix for a GTC code $\mathcal{C}(A,D)$ is given by $H=I_n \otimes A - (D^t \otimes I_n)(A^t \otimes I_n)$, where $\otimes $ denotes the Kronecker product, $A^t$ is the transpose of the matrix $A$, and $ I_n $ is the identity matrix of order $ n $. 
\end{prop}

\begin{proof}
If we take $ B = AD $ and $ C= O $ in Theorem 27.5.1 of \cite{prasolov1994}, the theorem follows easily.
\end{proof}

\begin{thm}
Let, $ A,D \in \mathbb{F}_q^{n \times n}$ and $ O $ be the null matrix of order $ n $. Then the following are true:
\begin{enumerate}
\item[a.] $A \in \mathcal{C}(A,D)$ if and only if $D=I_n$ or $A^2=O$.
\item[b.] If $ D $ is invertible then $ B \in \mathcal{C}(A,D) \Leftrightarrow A \in \mathcal{C}(B,D^{-1})$.
\item[c.] For $ A \neq O $, we have $I_n \in \mathcal{C}(A,D) \Leftrightarrow D=I_n$.
\end{enumerate}
\end{thm}

\begin{thm}
If $O \neq A \in \mathbb{F}_q^{n \times n}$ and $ D \neq I_n $, then the dimension $dim(\mathcal{C}(A,D)) \leqslant n^2-1$.
\end{thm}

\begin{proof}
It is clear from context of linear algebra that the dimension of $ \mathcal{C}(A,D) $ is at most $ n^2 $. Now, if $ dim(\mathcal{C}(A,D)) = n^2 $, then every matrix B satisfies the relation $ AB=BAD $. But, if we take $ B=I_n $ then $ A=AD $. Which is not possible for any $ D \neq I_n$. Hence, $ dim(\mathcal{C}(A,D)) \leq  n^2-1 $.
\end{proof}

\begin{thm}
Let $ A, D \in \mathbb{F}^{n \times n} $, and the GTC code $ \mathcal{C}(A,D) $ contains an invertible matrix, then

\centering $ dim(\mathcal{C}(A,D)) = dim(\mathcal{C}(A)) $.
\end{thm}

\begin{proof}
Let $ B $ be an invertible matrix in $ \mathcal{C}(A,D) $. Then, $ AB = BAD \Rightarrow A = BADB^{-1} $. Now, consider the linear mapping $ f_B : \mathcal{C}(A) \rightarrow \mathcal{C}(A,D) $ such that $ f_B(X) = XB $. This mapping is closed since, $ X \in \mathcal{C}(A) \Rightarrow AX = XA \Rightarrow AXB = XAB \Rightarrow A(XB) = (XB)AD~~(\because AB = BAD) \Rightarrow XB \in \mathcal{C}(A,D) $. Clearly, the mapping $ f_B $ is injective and hence we can conclude that $ dim(\mathcal{C}(A)) \leq dim(\mathcal{C}(A,D)) $.

Again, the mapping $ \phi_B : \mathcal{C}(A,D) \rightarrow \mathcal{C}(A) $ such that $ \Phi_B(Y) = YB^{-1} $ is closed since $ Y \in \mathcal{C}(A,D) \Rightarrow AY = YAD \Rightarrow AYB^{-1} = YADB^{-1} \Rightarrow AYB^{-1} = YB^{-1}BADB^{-1} \Rightarrow AYB^{-1} = YB^{-1}A ~~ ( \because A = BADB^{-1}) \Rightarrow YB^{-1} \in \mathcal{C}(A) $. The mapping $ \Phi_B $ is also injective. So, $ dim(\mathcal{C}(A)) \geq dim(\mathcal{C}(A,D)) $.

Combining both results we have $ dim(\mathcal{C}(A,D)) = dim(\mathcal{C}(A)) $.
\end{proof}
\begin{thm}
For all $D \in \mathbb{F}_q$, the code $\mathcal{C}(A,D)$ contains the product code $Ker(A) \otimes Ker(D^t A^t) $. If $Ker(A)$, $Ker(D^t A^t)$ have respective parameters $[n, k, d]$ and $[n, k', d']$, then $ \mathcal{C}(A,D) $ has parameters $ [n^2,K,D] $ with $ K \geq kk'$ and $ D \leq dd' $.

\end{thm}

\begin{proof}
Let, $ u \in Ker(A) $ and  $v \in Ker(D^t A^t) $. Then $ A(u v^t) = (Au)v^t = O $ and $ uv^t AD = u (D^t A^t v) = O $. Which shows that $ B = u v^t \in \mathcal{C}(A,D) $ as $ AB = O = BAD $. So, $  Ker(A) \otimes Ker(D^t A^t)  \subseteq \mathcal{C}(A,D) $.

The next part of the theorem will be established by a simple property of product code in \cite{macwilliams1977}.
\end{proof}

\section{Main results}
Let $ E, A \in  \mathbb{F}_q^{n \times n} $. We define a set $ \mathcal{T}_E := \lbrace B \in \mathbb{F}^{n \times n}: AB = BAD = E \rbrace $. According to our definition, we have
$ \mathcal{C}(A,D) = \bigcup\limits_{E \in \mathbb{F}^{n \times n}_q} \lbrace B: AB = BAD = E \rbrace =  \bigcup\limits_{E \in \mathbb{F}^{n \times n}_q} \mathcal{T}_E $.
Throughout this section we denote $ r_A $ as the rank of a matrix $ A $. 
\begin{prop}\label{p1}
Dimension of $ \mathcal{T}_O $ is less than or equal to $ n^2-n \cdot r_A $.
\end{prop}
\begin{proof}
According to our notation $ \mathcal{T}_O:= \lbrace B: AB = BAD = O \rbrace $. Now, $ AB = O \Rightarrow B \in Ker(A) $. Now let $\mathcal{K}_{AD,O} := \lbrace B: BAD = O \rbrace $. Since $ BAD = O \Rightarrow D^t A^t B^t = O \Rightarrow B^t \in Ker(D^t A^t) $. We take a mapping $ \psi: \mathcal{K}_{AD,O} \rightarrow Ker(D^t A^t) $ such that $ \psi (B) = B^t $. Clearly, this map is bijective. So, $ |\mathcal{K}_{AD,O}| =|Ker(D^t A^t)|  \Rightarrow dim(\mathcal{K}_{AD,O}) = dim(Ker(D^t A^t))  \Rightarrow dim(\mathcal{K}_{AD,O}) = n^2 - n \cdot r_{D^t A^t} = n^2 - n \cdot r_{AD}$. Again $ dim(Ker(A)) = n^2 - n \cdot r_A $. Now, $ dim(\mathcal{T}_O) \leq min \lbrace dim(Ker(A)),dim(k_{AD,O}) \rbrace = min \lbrace n^2 - n \cdot r_A, n^2 - n \cdot r_{AD} \rbrace = n^2 - n \cdot r_A$.
\end{proof}
\begin{prop}\label{p2}
If $ \mathcal{T}_E $ is non-empty then $ |\mathcal{T}_E| = |\mathcal{T}_O| $.
\end{prop}

\begin{proof}
Let $ B \in \mathcal{T}_O $ and  $ S_1 \in \mathcal{T}_E $ for a fixed $ E \neq O $ and $ S \in \mathbb{F}^{n \times n} $, then $ S_1 + B \in \mathcal{T}_E $. Which shows that the $ | \mathcal{T}_E | \geq   | \mathcal{T}_O| $. Let if possible $ | \mathcal{T}_E | > | \mathcal{T}_O| $. Then there exists an element $ S_2 \neq S_1 $ which is not of the form $ S_1 + B $ for  $ B \in \mathcal{T}_O $. Now, $ S_1, S_2 \in \mathcal{T}_E  \Rightarrow  S_2-S_1 \in \mathcal{T}_O $. But, $ S_2 = S_1 + (S_2 - S_1) $ is in the form of $ S_1 + B $, contradicts our assumption. Hence $ |\mathcal{T}_E| = |\mathcal{T}_O| $ is proved.
\end{proof}

\begin{thm}\label{t1}
Let $ A, D \in \mathbb{F}^{n \times n}_q $, then the GTC code $ \mathcal{C}(A,D) $ has the dimension less than or equal to $ n^2 - n \cdot r_A + n \cdot r_{AD}$.
\end{thm}

\begin{proof}
Let us consider $ \mathcal{T}_E := \lbrace B \in \mathbb{F}^{n \times n}_q: AB = BAD = E \rbrace $. Basically the set $ \mathcal{T}_E $ is the common solutions of $ AB = E $ and $ BAD = E $. Now, $ AB = E $ is possible if columns of $ B \in Columnspace(A) $ and $ BAD = E $ is possible if rows of $ B \in Rowspace(AD) $. So, $ B \in \mathcal{T}_E $ if $ B \in Columnspace(A) \cap Rowspace(AD) $. We denote, $ G = Columnspace(A) \cap Rowspace(AD) $. $ \mathcal{C}(A,D) = \bigcup\limits_{E \in G} \lbrace B: AB = BAD = E \rbrace =  \bigcup\limits_{E \in G} \mathcal{T}_E $. By Proposition \ref{p2} if $ \mathcal{T}_E $ is non empty then, it has the same cardinality as $ \mathcal{T}_O $. Let us assume $ \mathcal{T}_E $ is solvable and non empty for each $ C \in G $. Then,\\  \[ \mathcal{C}(A,D) = \bigcup\limits_{E \in G} \mathcal{T}_E   \Rightarrow |\mathcal{C}(A,D)| \leq |\mathcal{T}_O| \cdot |G|\] \[ \Rightarrow dim( \mathcal{C}(A,D) ) \leq dim(\mathcal{T}_O) + dim(G). \] Now, \[ dim(G) = n \cdot dim(Columnspace(A) \cap Rowspace(AD))\] \[ \leq n \cdot min \lbrace dim(Columnspace(A)), dim(Rowspace(AD)) \rbrace \] \[= n \cdot min \lbrace r_A, r_{AD} \rbrace = n \cdot r_{AD}.\] \\ Using Proposition \ref{p1} we have  $ dim(\mathcal{C}(A,D)) \leq  n^2 - n \cdot r_A + n \cdot r_{AD} $.
\end{proof}
\begin{cor}
If $ AD = O $ in Theorem \ref{t1} then $ dim(\mathcal{C}(A,D)) \leq  n^2 - n \cdot r_A $.
\end{cor}


Let $\Gamma$ be a graph with vertices $v_1, v_2, \dots, v_n$. The adjacency matrix of $\Gamma$ is a square matrix of order $n$ whose $(i,j)$-entry is $1$ if the vertices $v_i$ and $v_j$ are adjacent, otherwise the entry is $0$. Automorphism group of graph is the set of all automorphisms from the vertex set to itself of the graph which preserves adjacency. It is denoted by $Aut(\Gamma)$.
\begin{thm}
If $ A \in {F}^{n \times n}_q $ is the adjacency matrix of graph $ \Gamma_{1} $ and $ G_1 = Aut(\Gamma_{1}) $ and if $ AD \in {F}^{n \times n} $ is the adjacency matrix of graph $ \Gamma_2 $  and $ G_2= Aut(\Gamma_2) $ then the direct product $ G_1 \times G_2 $ acts on the code $ \mathcal{C}(A,D) $ by coordinate permutations.
\end{thm}

\begin{proof}
It is known by \cite{biggs1993} that a permutation matrix $ P $ lies in $ Aut(\Gamma) $ if and only if $ AP^{-1} = P^{-1}A $. Let $ (P,Q) \in  G_1 \times G_2 $ and $ B \in \mathcal{C}(A,D) $. Then we have, $ AB = BAD $, $ P^{-1}A = AP^{-1} $, and $ QAD = ADQ $. So, $ AB = BAD \Rightarrow P^{-1}AB = P^{-1}BAD \Rightarrow P^{-1}ABQ = P^{-1}BADQ \Rightarrow AP^{-1}BQ = P^{-1}BQAD \Rightarrow P^{-1}BQ \in \mathcal{C}(A,D) $. So, $$\begin{array}{cccc} \phi: & (G_1 \times G_2) \times \mathcal{C}(A,D) & \rightarrow & \mathcal{C}(A,D) \\  & (P,Q) \times B & \mapsto & P^{-1}BQ \end{array}$$ is a group action. Hence the theorem is proved.
\end{proof}

In twisted centralizer codes \cite{Alahmadi2017235}, it was shown that for $a \neq 0, 1$, minimum distance can be larger than $n$. Here we show by few examples that GTC codes over binary fields have minimum distances greater than $n$ which is not possible to the twisted centralizer codes. Examples of optimal binary GTC codes are given below whose minimum distances are larger than the order of $A$.

\begin{example}
Suppose $ A = \begin{bmatrix}
1 & 1 & 0 \\
0 & 1 & 1 \\
1 & 1 & 0
\end{bmatrix} \in \mathbb{F}^{3 \times 3}_2 $ and $ D = \begin{bmatrix}
1 & 1 & 0 \\
1 & 1 & 1 \\
0 & 1 & 1
\end{bmatrix} \in \mathbb{F}^{3 \times 3}_2 $. Then an optimal binary GTC code $ [9, 2, 6] $ is acquired. 
\end{example}

\begin{example}
An optimal GTC code $ [9, 3, 4] $ is obtained for $ A = \begin{bmatrix}
0 & 1 & 0 \\
1 & 1 & 1 \\
0 & 1 & 0
\end{bmatrix} \in \mathbb{F}^{3 \times 3}_2 $ and $ D = \begin{bmatrix}
1 & 1 & 1 \\
1 & 1 & 1 \\
1 & 1 & 1
\end{bmatrix} \in \mathbb{F}^{3 \times 3}_2 $. 
\end{example}

In both examples, the order of $A$ is $3$. In first example we have minimum distance $6$ and in second it is $4$.

\section{Encoding-decoding procedure}
Let the generalized centralizer code $\mathcal{C}(A,D)$ has length $ n^2 $ and dimension $ k $.  $ \mathcal{C}(A,D) $ is a vector space over $ \mathbb{F}_q $ and it has a basis of dimension $ k $. Let  $ \lbrace B_1, B_2, \dots, B_k \rbrace $ is a basis of $ \mathcal{C}(A,D) $. So, for an information message $ (a_1, a_2, \dots, a_k) \in \mathbb{F}_{q}^{k}$ can be encoded as $ a_1 B_1+a_2 B_2+\dots+a_k B_k $.

The method of decoding is the reverse process of encoding. A receiver should know the basis $ \lbrace B_1, B_2, \dots, B_k \rbrace $ to decode the received message into information message.

We have already established that $ \mathcal{C}(A,D) $ is a linear subspace of $ \mathbb{F}_q^{n \times n} $ and hence we can state that $ \mathcal{C}(A,D) $ is an additive subgroup of $ \mathbb{F}_q^{n \times n} $. Then, cosets of $ \mathcal{C}(A,D) $ is in $ \mathbb{F}_q^{n \times n} $. We can use $ A $ as a parity check matrix since $AB - BAD = O$ for every $ B \in \mathcal{C}(A,D) $. To decode the information message, we can use syndrome decoding.

\begin{defn}
Let $ \mathcal{C}(A,D) $ be the non-empty generalized twisted centralizer code for a matrix $ A $ twisted by the matrix $ D $ and let $ B \in \mathcal{C}(A,D) $. The \textit{syndrome} of B is defined a \centering $S_A(B)= AB-BAD$.
\end{defn}

\begin{thm}
Consider two matrices $B_1, B_2 \in \mathbb{F}_q^{n \times n}$. Then $S_A(B_1) = S_A(B_2)$ iff $B_1$ and $B_2$ are in same coset of $\mathcal{C}(A,D)$.
\end{thm}

\begin{proof} Let \[ S_A(B_1) = S_A(B_2)\]  \[\Leftrightarrow  A B_{1} - B_{1} A D = A B_{2} - B_{2} A D \]   \[\Leftrightarrow A(B_1-B_2) = (B_1-B_2)AD \] \[ \Leftrightarrow B_1- B_2 \in \mathcal{C}(A,D).\]

Therefore, $ B_1 $ and $ B_2 $ are in same coset of $ \mathcal{C}(A,D) $ iff $ B_1- B_2 \in \mathcal{C}(A,D)$. Hence the theorem is proved.
\end{proof}

The motive of defining syndrome is that syndrome computation is more easier than the computation by using $ n^2 \times n^2 $ parity check matrix for authenticity. By using the parity-check matrix we need $ O(n^4) $ multiplicative complexity but using $ A $ as a parity check matrix computation has $ O(n^m) $ complexity where $ m < 2.3729 $ by \cite{Williams2012}. Also, the purpose of taking the code $\mathcal{C}(A,D):= \lbrace B \in \mathbb{F}_q^{n \times n}|AB=BAD\rbrace $ instead of $\mathcal{C}(A,D):= \lbrace B \in \mathbb{F}_q^{n \times n}|AB=DBA)\rbrace $ is due to less computational complexity of the syndrome because at that time we can consider a fixed matrix $C=AD$. The process of syndrome decoding is very similar with usual decoding process.

\section{Some optimal generalized twisted centralizer codes}
Here we provide some examples of optimal twisted centralizer codes where the optimality is verified by \cite{Grassl2007}.

\begin{example}
Let $ A = \begin{bmatrix}
1 & 1 & 1 \\
1 & 1 & 1 \\
1 & 1 & 1
\end{bmatrix} \in \mathbb{F}^{3 \times 3}_2 $. Then, the matrices $ D_1$, $D_2$, $D_3 \in \mathbb{F}^{3 \times 3}_2 $ where
$ D_1 = \begin{bmatrix}
1 & 1 & 1 \\
1 & 0 & 1 \\
0 & 1 & 1
\end{bmatrix} $, $ D_2 = \begin{bmatrix}
1 & 0 & 1 \\
1 & 1 & 1 \\
0 & 0 & 1
\end{bmatrix} $, and $ D_3 = \begin{bmatrix}
1 & 1 & 0 \\
0 & 0 & 1 \\
1 & 1 & 1
\end{bmatrix} $ give optimal GTC codes $ [9, 5, 3] $,  $ [9, 4, 4] $ and $ [9, 6, 2] $ respectively.
\end{example}

\begin{example}
Consider $ A = \begin{bmatrix}
1 & 1 & 1 & 1 \\
1 & 1 & 1 & 1 \\
1 & 1 & 1 & 1 \\
1 & 1 & 1 & 1
\end{bmatrix} \in \mathbb{F}^{4 \times 4}_2 $. Optimal GTC codes $ [16, 9, 4] $, $ [16, 10, 4] $ and $ [16, 12, 2]$ are obtained for $D_1$, $D_2$ and $D_3 \in \mathbb{F}^{4 \times 4}_2$ respectively, where $ D_1 = \begin{bmatrix}
1 & 1 & 1 & 1 \\
1 & 1 & 0 & 1 \\
0 & 1 & 1 & 1 \\
1 & 0 & 1 & 1
\end{bmatrix} \in \mathbb{F}^{4 \times 4}_2 $, $ D_2 = \begin{bmatrix}
1 & 1 & 1 & 1 \\
1 & 1 & 0 & 0 \\
1 & 1 & 1 & 0 \\
1 & 0 & 1 & 1
\end{bmatrix} \in \mathbb{F}^{4 \times 4}_2 $ and $ D_3 = \begin{bmatrix} 
1 & 1 & 1 & 1 \\
1 & 1 & 0 & 1 \\
0 & 0 & 0 & 1 \\
0 & 0 & 1 & 1
\end{bmatrix} $.
\end{example}

\begin{example}
For $ A = \begin{bmatrix} 1 & 2 & 2 \\ 2 & 1 & 1 \\ 2 & 1 & 1 \end{bmatrix} $ and $ D = \begin{bmatrix} 1 & 1 & 1 \\ 1 & 1 & 2 \\ 1 & 2 & 1 \end{bmatrix} $ GTC code $ [9, 5, 4] $, and for $ A = \begin{bmatrix} 1 & 0 & 1 \\ 1 & 1 & 1 \\ 2 & 1 & 0 \end{bmatrix} \in \mathbb{F}^{3 \times 3}_3 $ and $ D = \begin{bmatrix} 0 & 0 & 1 \\ 0 & 1 & 1 \\ 2 & 2 & 2 \end{bmatrix} \in \mathbb{F}^{3 \times 3}_3 $  GTC code $ [9, 3, 6] $ are attained.
\end{example}

\begin{example}
For $ A = \begin{bmatrix} 2 & 2 & 2 & 2 \\ 2 & 2 & 2 & 2 \\ 2 & 2 & 2 & 2 \\ 2 & 2 & 2 & 2 \end{bmatrix} \in \mathbb{F}^{4 \times 4}_3 $ and $ D = \begin{bmatrix} 1 & 0 & 2 & 1 \\ 1 & 2 & 0 & 1 \\ 0 & 0 & 2 & 0 \\ 2 & 0 & 0 & 1 \end{bmatrix} \in \mathbb{F}^{4 \times 4}_3 $ GTC code $ [16, 10, 4] $, and for $ A = \begin{bmatrix} 0 & 2 & 1 & 0 \\ 1 & 1 & 2 & 1 \\ 1 & 0 & 2 & 1 \\ 1 & 0 & 0 & 2 \end{bmatrix} \in \mathbb{F}^{4 \times 4}_3 $ and $ D = \begin{bmatrix} 1 & 0 & 0 & 1 \\ 0 & 2 & 0 & 2 \\ 0 & 0 & 2 & 0 \\ 2 & 0 & 0 & 1 \end{bmatrix} \in \mathbb{F}^{4 \times 4}_3 $ GTC code $ [16, 3, 10] $ are obtained.
\end{example}

\section{Generalized centralizer codes of length less than $n^2$}
Generalized twisted centralizer codes of length less than $n^2$ can be constructed by choosing $B$ with entries are $0$ in fixed $i$ positions. Clearly, that set is a subcode of $\mathcal{C}(A,D)$ say $\mathcal{S}(A,D)$. Puncturing those fixed $i$ number of entries of the subcode  $\mathcal{S}(A,D)$, we get a new code of length $n^2 - i$ where $i$ is a positive integer in $1 \leq i <(n-1)^2$.

Using this puncturing method, the multiplicative complexity is reduced as $ O((n-i)^m) $ where $ m <2.3729 $. This result works well when $ i $ is near to $ 1 $.

\begin{example}
For $ A = \begin{bmatrix}
1 & 1 & 1 \\
1 & 2 & 2 \\
0 & 0 & 1
\end{bmatrix} \in \mathbb{F}^{3 \times 3}_3 $ and $ D = \begin{bmatrix}
0 & 0 & 2 \\
1 & 1 & 0 \\
1 & 2 & 2
\end{bmatrix} \in \mathbb{F}^{3 \times 3}_3 $ an optimal code $ [7, 2, 5] $ is obtained by puncturing $(1,2)$-entry and $(2,3)$-entry in $B$.
\end{example}

\begin{example}
For $ A = \begin{bmatrix}
0 & 0 & 1 & 1 \\
0 & 1 & 0 & 0 \\
1 & 1 & 1 & 0 \\
1 & 0 & 0 & 1
\end{bmatrix} \in \mathbb{F}^{4 \times 4}_2 $ and $ D = \begin{bmatrix}
1 & 0 & 1 & 1 \\
1 & 1 & 0 & 0 \\
0 & 0 & 1 & 1 \\
0 & 1 & 0 & 1
\end{bmatrix} \in \mathbb{F}^{4 \times 4}_2 $ an optimal code $ [12, 3, 6] $ is obtained by puncturing all entries in $4^{th}$ column of $B$.
\end{example}

\section{Conclusion}
In this paper we have generalized the idea of twisted centralizer codes \cite{Alahmadi2017235}. It has shown in Section $2$ that the dimension of a twisted centralizer code is equal to GTC code. In Section $3$, an upper bound on dimension of GTC code has been derived. Encoding and decoding procedure has been implemented to GTC codes. Length of centralizer codes could be shorten by using concept of puncture codes. In twisted centralizer codes, minimum distance can be at most $ n $ when the field is binary whereas we have constructed GTC code with minimum distance more than $ n $ when $ D \in \mathbb{F}_2^{n \times n} $. An error can be corrected by simply looking at the syndrome. But finding $t$-errors ($t > 1$) in twisted centralizer codes or in GTC codes is still open.\\

\textbf{Acknowledgements}
 
The authors Joydeb Pal is thankful to DST-INSPIRE and Pramod Kumar Maurya is thankful to MHRD for their financial support to pursue his research work.


\section*{References}


\end{document}